\documentclass[conference]{IEEEtran}
\makeatletter
\def\ps@headings{%
\def\@oddhead{\mbox{}\scriptsize\rightmark \hfil \thepage}%
\def\@evenhead{\scriptsize\thepage \hfil \leftmark\mbox{}}%
\def\@oddfoot{}%
\def\@evenfoot{}}
\makeatother
\ifCLASSINFOpdf
\else
\fi

\usepackage{psfrag}
\usepackage{amsfonts,amssymb,amsmath}
\usepackage{amsbsy}
\usepackage{mathrsfs}
\usepackage{dsfont}
\usepackage{graphicx}
\usepackage{url}
\usepackage{paralist}
\usepackage[usenames,dvipsnames]{color}
\usepackage{subfigure}
\usepackage{hyperref}
\usepackage{algorithm}

\newcommand \defn {\mathrel{\triangleq}}
\newcommand\set[1] {\left\{#1\right\}}

\renewcommand\Pr{\mathop{\mbox{\ensuremath{\mathbb{P}}}}\nolimits}
\newcommand\Ex{\mathop{\mbox{\ensuremath{\mathbb{E}}}}\nolimits}

\newcommand{\E}{$\mathcal{E}$}

\newcommand{\U}{\mathcal{U}}

\newcommand{\Lo}{\mathcal{L}}

\newcommand {\bxt}{\mathbf{x}_t}

\newcommand{\Q}{\mathcal{Q}}

\newcommand{\Ho}{H}

\newcommand \LQ {\ell_\mathcal{Q}}
\newcommand \GU {g_\mathcal{U}}

\newtheorem{assumption}{Assumption}{\bfseries}{\itshape}
\newtheorem{lemma}{Lemma}{\bfseries}{\itshape}
\newtheorem{definition}{Definition}{\bfseries}{\itshape}
\newtheorem{theorem}{Theorem}{\bfseries}{\itshape}

\newcommand \loss {L}

\newcommand\Naturals {{\mathds{N}}} 

\newcommand \HIPER {\textsc{HiPER}}

\begin{document}
%
\title{Near-Optimal Blacklisting}

\author{\IEEEauthorblockN{Christos Dimitrakakis}
\IEEEauthorblockA{EPFL \\Lausanne,
Switzerland\\
Email: christos.dimitrakakis@epfl.ch}
\and
\IEEEauthorblockN{Aikaterini Mitrokotsa}
\IEEEauthorblockA{Chalmers University of Technology\\
Gothenburg, Sweden\\
Email: mitrokatkm@gmail.com}
}


%


\maketitle

\begin{abstract}
Many applications involve agents sharing a resource, such as networks or services. When agents are honest, the system functions well and there is a net profit. Unfortunately, some agents may be malicious, but it may be hard to detect them. We consider the intrusion \emph{response} problem of how to permanently \emph{blacklist} agents, in order to maximise expected profit. This is not trivial, as blacklisting may erroneously expel honest agents. Conversely, while we gain information by allowing an agent to remain, we may incur a cost due to malicious behaviour. We present an efficient algorithm (\HIPER) for making near-optimal decisions for this problem. Additionally, we derive three algorithms by reducing the problem to a Markov decision process (MDP). Theoretically, we show that \HIPER{} is near-optimal. Experimentally, its performance is close to that of the full MDP solution, when the (stronger) requirements of the latter are met.
\end{abstract}


%
\IEEEpeerreviewmaketitle

\section{Introduction}

%

We consider the \emph{decision making} problem of blacklisting potentially malicious nodes or agents that share a resource or network based on partial information. As motivation, consider a communication network which is monitored by a network management system. Nodes can be of one of two types: malicious (e.g. dropping or corrupting packets, creating undue congestion), or honest. At each time-step (e.g. reporting period), we get a set of readings, giving some information about the behaviour of each node during that period. We want to optimally decide whether to blacklist a node, or maintain it in the system for one more time-step. 

In order to make the problem non-trivial, we assume that for every honest node in the network, we have some fixed tangible gain at each time period. This would be the case, if all participation was done through a subscription model as in internet service providers (ISPs). On the other hand, we incur a (hidden) cost for each malicious node that participates. Thus, it is in our interests to kick out malicious nodes as soon as possible, but never to expel honest ones. 

We should emphasize that this is \emph{not} an intrusion detection problem. In fact, the readings that we obtain for each node could be seen as the output of some intrusion detection system (IDS). Rather, we are more concerned about the decision making aspect: what is the \emph{optimal response} to the IDS outputs, given assumptions about the cost of malicious behaviour?

This setting of keeping suspicious nodes in the network until we become more certain about their type appears in many applications such as:
\begin{inparaenum}
\item blacklisting clients of an ISP
\item shutting down malware-infected hosts in an internal network
\item expelling selfish nodes from a peer-to-peer network. 
\end{inparaenum}
In all of the above cases, any single piece of information is not enough to condemn a node to blacklisting. Rather, a sufficient amount of statistics has to be collected before we are sure that removing a node is more beneficial than keeping it. In this paper, we propose and consider a number of algorithms for tackling this problem in a general setting.

More precisely, our first contribution is a decision-theoretic approach based on distribution-free high probability bounds. The bounds require very little prior information and can be used to trade off the cost of removing honest nodes with that of keeping malicious nodes in the network for too long. We prove that this High Probability Efficient Response algorithm (\HIPER{}) has low {\em worst-case expected loss} relative to an oracle which knows the type of every node.

Our second contribution is a set of Bayesian decision-theoretic approaches that we derive by formalising the problem as a Markov decision process (MDP). These require some further assumptions. In particular, it is necessary to fully specify a structure and prior parameters for the underlying statistical model. In addition, making optimal decisions according to such models is computationally intractable. Consequently, we consider some approximate algorithms. Of these, an optimistic approximation has similar performance to that of \HIPER{}, while a finite lookahead approximation has increased performance, at the cost of additional computation.

The paper is organised as follows. In the remainder of this section we give some background, present related work and our contributions. Section \ref{sec:prelim} introduces notation while Section \ref{sec:expectedloss} specifies the loss model. Section \ref{sec:hiper} presents the proposed \HIPER{} algorithm as well as the bounds on the {\em worst-case expected loss}. Section \ref{sec:POMDP} describes the decision-theoretic approaches which model the problem as an MDP and are used in the performance comparisons with the \HIPER{} algorithm
 while Section \ref{sec:experiments} describes the evaluation experiments.
Finally, Section \ref{sec:conclusion} concludes the paper. The appendix provides proofs of technical lemmas and some useful auxiliary results.

\subsection{Background}
\label{sec:related}

The problem we consider falls within decision theory. In particular, the scenario we investigate can be reduced to the {\em optimal stopping} problem~\cite{groot}, which can be modelled as an MDP ~\cite{groot} or as a (potentially unknown) partially observable MDP (POMDP)~\cite{POMDP}.

More precisely, in our setting, the nodes can be one of two types: honest or malicious. However, we initially start out without knowing what type each node is. Consequently, we must gather data (observations) to reduce our uncertainty about their types. Unfortunately, we can only do so while a node remains within the network. However, the longer we maintain a malicious node in the network, the more loss we incur. Conversely, once we remove an honest node, we will obtain no more profit from it. So, the problem can be reduced to deciding at what time, or under which conditions, to remove a given node from the network, if at all. Thus, our scenario can be seen as a type of {\em optimal stopping} problem.

The stopping problem has been extensively studied in general~\cite{groot}, while partial monitoring games in general have also received a lot of attention recently~\cite{CesaBianchi-Lugosi:PLG}. However, to the best of our knowledge, the general hidden reward stopping problem has not been previously studied in the literature. On the other hand, the specific application we consider can be seen as a type of \emph{optimal intrusion response}.

Most of the previous research on intrusion response has concentrated on the POMDP formalism. Indicative publications are those by Zonouz {\em et al.} \cite{zonouz}, Zan {\em et al.}, \cite{zan} and Zhang {\em et al.} \cite{zhang}, which have all proposed an intrusion response through modelling the process as a POMDP~\cite{POMDP}. More precisely, Zonouz {\em et al.} \cite{zonouz} proposed a Response and Recovery Engine (RRE) based on a game-theoretic response strategy against adversaries modelled as opponents in a non-zero-sum, two-player Stackelberg stochastic game. In each step of the game  RRE chooses the response actions using an approximate POMDP solver. More precisely, using the most likely state (MLS) \cite{MLS} approximation, the  POMDP is converted to a competitive  Markov Decision Process (MDP), which is then solved using a look ahead search (i.e. approximate planning).
Zhang {\em et al.} use the POMDP to integrate low level IDS alerts with high level system states, while Zan {\em et. al.} \cite{zan} propose to solve the intrusion response problem as a factored POMDP model. Additionally, they decompose the POMDP into small sub-POMDPs and compute the response policy using the MLS approximation technique.  However, in our case MLS as an approximation is too crude to be used, since it would essentially result in a completely random policy, as there are only two possible hidden states each node can be in. An entirely different approach, policy-gradient methods,
is employed by~\cite{dejmal2008reinforcement} in the context of combating denial-of-service attacks in P2P networks. However, this approach requires observing the rewards, which are in fact hidden in our case.

\subsection{Our contributions}
Our first proposed algorithm relies on bounds which do not require knowledge of prior probabilities regarding the type of a node (honest or malicious) neither known distributions for the observations corresponding to honest or malicious nodes. We only need to know the mean of each of these distributions. Consequently, it is substantially more lightweight than MDP solvers, since we take decisions without performing explicit planning. Thus, it is more suitable for resource constrained environments. We analyse the {\em expected loss} of this algorithm, and show that it is not significantly worse to that of an oracle which already knows each node's type.

Our second contribution is to derive three approximate MDP solvers. In contrast to previous work, in our scenario the reward is {\em never observed} by the algorithm.\footnote{Although of course the reward is used in the experiments to measure performance.} This corresponds to reality, since we frequently do not know which nodes give us negative rewards.\footnote{Conversely, if we could observe the rewards, it would be trivial to identify malicious nodes.} Furthermore, two of our MDP algorithms are different from those previously employed in the intrusion response literature, as we forego the most-likely-state approximation commonly used in POMDP approaches.  We first consider a myopic approximation.\footnote{This is equivalent to the most likely state approximation and to a sequential probability ratio test under some conditions.} The second approach is a lightweight {\em optimistic} approximation that performs no planning, which is derived from upper bounds~\cite{dimitrakakis-icaart2010} on Bayesian decision making in unknown MDPs~\cite{duff2002olc}. To our knowledge, this approach has not been used in similar problems before. Finally, we consider online planning with finite lookahead~\cite{RossPineau:OnlinePlanningPOMDPs:jmlr2008,groot}. This approach takes decisions which consider the impact of all our possible future actions up to some horizon. This approach has been employed in other applications such as dialogue modelling~\cite{bui2006tractable}, autonomous underwater vehicle mapping~\cite{saigol2009information}, preference elicitation~\cite{boutilier2002pomdp} and sensor scheduling~\cite{he2004sensor} in wireless sensor networks.

\section{Preliminaries}
\label{sec:prelim}
We consider a set of nodes, which can be either honest or malicious. We assume there is a reliable way to obtain statistics from each node, such as an IDS that gives us a numerical score for each node. We denote by $\Q$ the set of all malicious nodes and by $\U$ the set of all honest nodes. We consider that there is an entity \E\ (for instance an Internet Service Provider (ISP) or a network administrator) who gains some reward $\GU$ for each moment that an honest node remains in the network and has a cost $\LQ$ for each moment that a malicious node stays in the network. A node may be removed by \E\ at any time, for example through black-listing. However, re-inserting a removed node is not normally possible.

We use $N$ to denote the (possibly random) time at which \E\ removes a node from the network. In addition, any honest node may leave the network at some (random) time $\Ho$. Specifically, we assume that an honest node may decide to leave the network with some small probability $\lambda > 0$, independently over time. Then it holds that $\Ex[\Ho] = \frac {1}{\lambda}$. 
\begin{assumption}
  We assume that each node has a fixed type (i.e. honest or malicious) that is not changing over time. The type is hidden from \E{}.
  \label{as:samestate}
\end{assumption}
This does not mean that a node cannot \emph{behave} maliciously during one period and honestly the next: a node that drops packets on purpose, might not do so all the time. Of course, \E\ not only does not know the type of each node, but it also never observes the rewards obtained or the cost incurred.

At each time-step $t$ and for each node $i$, \E\ receives an \emph{information} signal $x_{i,t}\in [0,1]$, characterising the behaviour of that node $i$ within the time interval $t \in \Naturals$.
This signal can be seen as the output from some IDS, summarising  the behaviour of that node during that period.
\begin{assumption}
  We assume that $x_{i,1}, \ldots, x_{i,t}$ are independent,
(but not identically) distributed, random variables and:
\begin{equation}
\Ex[x_{i,t}\mid \Q]=q, \quad \Ex[x_{i,t}\mid \U]=u.
  \label{eq:meandistribution}
  \end{equation}
\end{assumption}
While the expected value is constant for all $t$,
the observed average of $\frac{1}{t} \sum_{k=1}^t x_{i,t}$ for each
node $i$ will initially be far from the expected value for small $t$.
The average, together with the total number of observations for each
node form a summary of the information received by each node. The relationship between these quantities will be looked at more closely in the analysis.

Finally, we place no specific meaning to $q$ and $u$ in this work, as they are application-dependent. In an \emph{intrusion response} scenario (e.g. \cite{costSensitive}), they could be considered as the {\em detection rate} (DR) and the {\em false alarm rate} (FA) correspondingly of an employed intrusion detection system. Then $x_{i,t}$ would correspond to alarm signals, with lower and high values for innocent-looking and suspicious behaviour respectively. Correspodingly, in a {\emph peer-to-peer} scenario (e.g. \cite{si2009distributed}), they could be fairness or reputation scores of each node.

In the remainder, we always refer to some arbitrary node in the network and thus make no distinction between nodes. This is because the algorithms that we examine, consider each node \emph{independently} of the others. Consequently, the following section analyses the expected loss for a single node of unknown type.

\section{The loss model}
\label{sec:expectedloss}
As previously mentioned, \E\ obtains a small gain for each time-step an honest node is within the network, and a small loss for each time-step a malicious node remains in the network. Formally, we can write that the total gain $G$ we obtain from some node $i$, which \E\ removes at time $N$, and which would voluntarily leave at time $\Ho$ is:
\begin{equation}
  \label{eq:gain1}
  G(i, \Ho, N) =
  \begin{cases}
    - N \LQ, & i \in \Q\\
    \min\{\Ho, N\} \GU, & i \in \U.
  \end{cases}
\end{equation}
\E\ wants to choose some node removal policy $\pi$ that maximises his total expected gain. That means that \E\ needs to keep as many as possible honest nodes in the network and eliminate the nodes that behave maliciously. In our analysis, we compare the expected gain of our policy $\pi$ with that of an {\em oracle}. The oracle always knows the type of each node (i.e. honest or  malicious), and thus, employs the optimal policy $\pi^*$. For $i \in \Q$,  according to the optimal policy $\pi^{*}$ it holds $N = 0$, while for $i \in \U$ according to the optimal policy $\pi^*$  it is $N = \infty$. Correspondingly,
\begin{equation}
  \label{eq:gain2}
  \Ex_{\pi^*} [G(i)] =
  \begin{cases}
    0, & i \in \Q\\
    \Ex[\Ho] \GU, & i \in \U.
  \end{cases}
\end{equation}
Let the loss $L$ be the difference between the gain of the optimal policy and our policy.  In particular, the {\em expected loss} of policy $\pi$ for a node of type $v$ is defined as:
\begin{equation}
  \Ex_\pi [L \mid v] = \Ex_{\pi^*(v)}[G\mid v] - \Ex_\pi[G\mid v],
  \label{eq:expected-loss1}
\end{equation}
where the $i$ subscript has been dropped for simplicity. The expected loss is 
 bounded by the {\em worst-case expected loss}:
\begin{equation}
  \Ex_\pi [L] \leq  \max_{v\in\{\Q,\U\}} \Ex_\pi [L \mid v],
  \label{eq:expected-loss2}
\end{equation}
which we wish to minimise. If \E\ removes node $i$ from the network at random time $N$, then he does not receive any more observations $x_{i,t}$  for this node from the IDS.
Thus, in essence, we want to find a {\em stopping rule}, that will let \E\ to determine the random time $N$ at which stopping occurs, i.e. \E\ takes the decision that $i\in\Q$ and removes it from the network.
We note that, $0\leq N\leq\infty$, where $N=\infty$ if stopping never occurs.

Since \E\ does not know if node $i$ is honest or malicious, it must collect a sufficient number of samples so as to only remove nodes for which it is reasonably certain that they are malicious. On the other hand, malicious nodes must be removed as soon as possible, since the operator incurs a cost for every moment they remain in the network. The first algorithm we consider uses simple statistics to make nearly optimal decisions about which nodes to keep.

\section{The \HIPER{} Algorithm}
\label{sec:hiper}
The algorithm, depicted in Alg.~\ref{fig:Algorithm}, uses the knowledge we have about malicious and honest nodes (see equation \ref{eq:meandistribution}). This is done by
calculating the average of all the observations generated by a node $i$ until time $t$:
\begin{equation}
  \theta_t \defn \frac{1}{t} \sum_{k=1}^{t} x_{i,k},
  \label{eq:average}
\end{equation}
and adding an appropriate {\em confidence interval} so that errors are
made with low probability.  Informally, \HIPER{} keeps nodes in the
network as long as the statistic $\theta_t$ is sufficiently far from
the expected statistic $q$ of malicious nodes. In order to avoid
throwing away honest nodes prematurely, it always keeps nodes for a
certain number of steps to obtain more reliable statistics. However,
as time passes, it needs more and more evidence to kick a node
out. Consequently, the probability that an honest node is thrown out
is bounded.


The analysis of the algorithm proceeds in three steps. First, we calculate the expected loss of the algorithm when faced with a node of malicious type. Then, we calculate the loss for honest nodes. Subsequently, we combine the two losses and tune the algorithm's input parameters to obtain an overall loss bound.

\begin{algorithm}
  \centering{
    \scalebox{1.0}{
      \mbox{
        \centering \tabcolsep=5pt
        \begin{tabular}{l}
          \textbf{Parameters:} $\delta,\Delta,q\in[0,1]$\\
          \textbf{Loop:} \textbf{For} each node $i$ in the network:\\
          \hspace{1.3cm}\textbf{For} each time-step $t$ do:\\
          \hspace{1.7cm}\textbf{if }$|\theta_t -q|< \sqrt{\frac{\ln(2/\delta)} {2t}}$ and $t>\frac{\ln(2/\delta)}{2\Delta^2}$ \textbf{then}\\
          %
          %
          \hspace{2cm} remove node $i$ from the network\\
          \hspace{1.7cm}\textbf{else} keep node $i$ in the network.\\
          \hspace{1.7cm}\textbf{end if}\\
          \hspace{1.3cm}\textbf{end For}\\
          \hspace{1cm}\textbf{end For}
        \end{tabular}
      }
    }
  }
  \caption{\HIPER{} Algorithm for Optimal Response}
  \label{fig:Algorithm}
\end{algorithm}

%

The first bound only depends upon the input parameter $\delta$, the error probability we wish to accept, and the loss $\LQ$ incurred by malicious nodes. We prove that the expected loss is polynomially bounded in terms of both $\delta$ and $\LQ$.
\begin{lemma}
  For Algorithm~\ref{fig:Algorithm}, with input parameter $\delta$, and $\Delta = |u - q|$, the {\em expected loss} when the node is malicious
is bounded as:
  \begin{equation}
    \Ex[L\mid \Q]\leq \frac {\LQ} {(1-\delta)^2} 
  \end{equation}
  \label{lem:lossattacker}
\end{lemma}
The proof of this lemma can be found in the appendix.  Naturally, the
expected loss is linearly dependent on the loss of keeping a malicious
node in the network, while the dependence on the error probability is
quadratic.

The second bound depends on the input parameter $\Delta$, which corresponds to how far we expect the statistics of honest nodes to be from $q$, the gain obtained by honest nodes $\GU$ and the leaving probability of honest nodes $\lambda$. Once more, we obtain a polynomial loss bound in terms of those variables.
\begin{lemma}
  If $\Delta = |u - q|$, then the {\em expected loss} when the node is honest is bounded by:
  \begin{equation}
    \Ex [L\mid \U] \leq \frac {\GU(\Delta^2+2)} {\lambda (\Delta^2+2\lambda)}.
  \end{equation}
  \label{lem:lossuser}
\end{lemma}
The proof of this lemma can be found in the appendix.
Similarly to the previous lemma, there is a linear dependence on the loss that is incurred when we erroneously remove an honest node, and a quadratic dependence on the rate of departure. In addition, there is a weak dependence on the gap $\Delta$ between the two means.

Finally, we can combine everything in one bound by selecting a value for $\delta$ that depends on $\Delta$ and which simultaneously makes the bounds tight:
\begin{theorem} Set $\Delta = |u - q|$ and select:
  \begin{equation}
    \delta= 1- \sqrt{\frac  {\LQ\lambda(\Delta^2 +2\lambda)} {\GU(\Delta^2+2)}}
    \label{eq:delta}
  \end{equation}
  then the expected loss $\Ex \loss$ is bounded by:
  \begin{equation}
    \Ex (\loss) \leq \Lo_1 \defn  \frac {\GU(\Delta^2+2)} {\lambda (\Delta^2+2\lambda)} .
    \label{eq:threshold-loss-bound}
  \end{equation}
  \label{theorem:hiper}
\end{theorem}
\begin{proof}
  If we substitute (\ref{eq:delta}) in (\ref{eq:expected-loss1}) we get:
  \begin{equation}
    \Ex[L| \Q] \leq \frac{\LQ} {(1-\delta)^2} = \frac{\LQ} {\frac  {\LQ\lambda(\Delta^2 +2\lambda)} {\GU(\Delta^2+2)}} =\frac {\GU(\Delta^2+2)} {\lambda (\Delta^2+2\lambda)}.
  \end{equation}
  Thus, using \eqref{eq:expected-loss2} and Lemmas \ref{lem:lossattacker} and \ref{lem:lossuser} we get:
  \[
  \Ex (\loss) \leq \frac {\GU(\Delta^2+2)} {\lambda (\Delta^2+2\lambda)}
  \]
\end{proof}
This theorem shows that the performance of \HIPER{} only very weakly depends on the gap $\Delta$ between honest and malicious nodes. In addition, it is optimal up to a polynomial factor. 

\section{Markov decision Process approximations}
\label{sec:POMDP}

As mentioned in the introduction, our setting corresponds to an optimal stopping problems. As these can be modelled as MDPs~\cite{groot}, it may be useful to solve the problem directly using the MDP formalism.

To cast our problem in this setting, we need to specify: a) the {\em prior} probability for each node being {\em honest} or {\em malicious}; b) a {\em known} distribution family for the observation distribution, conditioned on whether the node under consideration is honest or malicious; c) a planning algorithm for calculating our responses. This can be quite demanding computationally, as the solution to the problem requires planning in a large tree. However, they can result in much better performance.

\subsection{Intrusion Response and POMDP}

A  Partially Observable Markov Decision Process (POMDP) \cite{POMDP} is a generalisation of a Markov Decision Process (MDP).
More precisely, a POMDP models the relationship between an agent and its environment  when the agent cannot directly observe the underlying state. A POMDP can be described as a tuple $<S, A, O, T, \Omega, R>$ where $S$ is a finite set of states, $A$ is a set of possible actions, $O$ is a set of possible observations, $T$ is a set of conditional transition probabilities and $\Omega$ is a set of conditional observation probabilities and $R: A, S\rightarrow \mathbb{R}$.

We can model our intrusion response problem as a POMDP if we consider that a node of the network at each time-step $t$ has a state $s_{t}\in S$ with $s_{t}=(v_{t},c_{t})$ where $v_{t}\in\{0,1\}$ and $c_{t}\in\{0,1\}$ such that:
\begin{equation*}
  v_{t}=
  \begin{cases}
    0, & \text{ if the node is honest},\\
    1, & \text{ if the node is malicious.}
  \end{cases}
\end{equation*}
\begin{equation*}
  c_{t}=
  \begin{cases}
    0, & \text{ if the node is in the network},\\
    1, & \text{ if the node is out of the network.}
  \end{cases}
\end{equation*} 
where it holds that $\Pr(v_{t+1}=v_{t})=1$ since $v_{t}$ is stationary (i.e. a malicious node is always malicious and an honest node remains honest)  based on Assumption \ref{as:samestate}.

Additionally, at each time-step $t$, \E\  can perform an action $a_{t}\in\{0,1\}$ such that:
\begin{equation*}
  a_{t}=
  \begin{cases}
    0, & \text{ if \E\ keeps the node in the network, }\\
    1, &\text{ if \E\ removes the node from the network.}
  \end{cases}
\end{equation*}
Furthermore, the following independence condition holds:
 $\Pr(v_{t+1}\mid v_{t},c_{t},a_{t})=\Pr(v_{t+1}\mid v_{t})$
since the type of a node (i.e. malicious or honest) does not depend on \E's action (i.e. remove from the network or not) neither on whether the node is in the network or out. In addition, since the type of a node never changes, it holds:
\begin{equation}
  \Pr(v_{t+1} = j \mid v_{t} = j) = 1.
  \label{eq:state1}
\end{equation}
Consequently, we remove the time subscript from $v$ in the sequel.
On the other hand the probability that a node will be in the network depends on if it is already in or out and the action that \E\ will take:
\begin{equation}
  \Pr(c_{t+1}\mid c_{t}, v, a_{t})=\Pr(c_{t+1}\mid c_{t}, a_{t})
  \label{eq:state2}
\end{equation}
From equations (\ref{eq:state1}) and (\ref{eq:state2}), it is evident that the POMDP under consideration is factored.  

To fully specify the model we must assume some probability distribution for the observations. Specifically, we model $x_t$ as drawn from a Bernoulli distribution\footnote{This distribution is particularly convenient for computational reasons, because closed-form Bayesian inference can be performed via the Beta conjugate prior~\cite{groot}. However, in principle it can be replaced with any other distribution family, without affecting the overall formalism.} with parameters $u$ and $q$ for honest and malicious nodes respectively:
$\Pr(x_{t}=1\mid v=0)=u \quad{\text  and } \quad \Pr(x_{t}=1\mid v=1)=q.$
Let $\bxt \defn (x_1, \ldots, x_t)$ be a $t$-length sequence of observations.
From Bayes' theorem, we obtain an expression for our {\em belief} at time $t$:
\begin{equation}
\Pr(v=j \mid \bxt)=\frac{\Pr(\bxt \mid v=j)\Pr(v=j)} {\sum_{i=0}^{1}\Pr(\bxt \mid v=i)\Pr(v=i)}
\label{eq:bayes-theorem}
\end{equation}
where $j\in\{0,1\}$.
Thus, the expected gain at time $t$ if \E\ decides to keep a node in the network is:
$\Ex[G_{t}\mid c_{t}=0, \bxt]=\Pr(v=0\mid \bxt)\cdot \GU- \Pr(v=1\mid \bxt)\cdot \LQ$
while the expected gain if \E\ decides to remove the node from the network is always:
$\Ex[G_{t}\mid c_{t}=1]=0.$
The problem is to find a policy $\pi : X^* \to A$, mapping from the set of all possible sequences of observations to actions, maximising the total  expected gain:
\begin{equation}
\Ex_\pi(G) = \Ex_\pi\left(\sum_{t=1}^\infty G_t\right).
\label{eq:expected-bayes-gain}
\end{equation}
Since future gains depend on any future observations we might obtain, the exact calculation requires enumerating all possible future observations. Consequently, the exact solution to the problem is intractable~\cite{groot,duff2002olc,dimitrakakis-icaart2010}. In the next section we describe possible approximations to this problem.

\subsection{POMDP algorithms}
We consider three algorithms:  a) A {\em myopic} algorithm, which only considers the expected gain at the current time-step; b) An {\em optimistic } algorithm, which computes an upper bound on the total expected gain; c) A {\em finite lookahead} algorithm, which performs complete planning up to some fixed finite depth. While these algorithms have appeared before in the general MDP literature, they have not been applied before to intrusion response problems. We do not consider the most likely state approximation (MLS), since in our case there are only two possible hidden states for a node, thus, rendering the approximation far too coarse for it to be effective.

\subsubsection{Myopic} In this case, \E\ only considers the expected gain for the next time-step when taking a decision. Consequently, \E\ keeps  the node in the network if:
 $\Ex[ G_{t}\mid a_{t}=0] > \Ex[ G_{t}\mid a_{t}=1].$
This algorithm is the closest to the MLS approximation among the ones considered. In fact, it is easy to see that it would be identical to MLS, as well as to a sequential probability ratio test, when $\LQ = \GU$.
\subsubsection{Optimistic} 
This rule constructs an upper bound on the value of the decision to keep a node in the network, which is based on Proposition~1 in~\cite{dimitrakakis-icaart2010}. Informally, this is done by assuming that the true type of the node will be revealed at the next time-step. Then \E\ keeps the node in the network if and only if:
 $ \Pr(v_{t}=0\mid \bxt)\cdot \GU / \lambda > \Pr(v_{t}=1\mid \bxt)\cdot \LQ.$
Intuitively, if the node is revealed to be malicious, then we can remove it at the next step and consequently we only lose $\LQ$. In the converse case, we can keep it for an expected $1/\lambda$ steps.

\subsubsection{Finite lookahead} 
The finite lookahead algorithm performs backwards induction~\cite{groot} up to some finite depth $T$, at every time-step. More precisely, any sequence of observations $\bxt = (x_1, \ldots, x_t)$ results in a posterior probability $\Pr(v_t \mid \bxt)$. Let: 
  $V_t \defn \sum_{k=t}^\infty G_k$
be the total gain starting from time-step $t$.  Then, the expected gain
under the optimal policy is determined recursively as follows:
\begin{align*}
  \label{eq:lookahead}
  \Ex(V_t \mid \bxt)
  &= 
  \max
  \{
  0,
  \Ex(G_t \mid \bxt, a_t = 0) + \Ex(V_{t+1} \mid \bxt)
  \}
  \\
  \Ex(V_{t+1} \mid \bxt)
  &=
  p_t \Ex(G_t \mid \bxt, x_{t+1} = 1)+\nonumber \\
  &(1 - p_t) \Ex(V_{t+1} \mid \bxt, x_{t+1} = 0)
\end{align*}
where $p_t \defn \Pr(x_{t+1} = 1 \mid \bxt) =  \sum_{i=0}^1 \Pr(x_{t+1} = 1 \mid v = i)  \Pr(v = i \mid \bxt)$ is the marginal posterior probability that $x_{t+1} = 1$. For more details on this backwards induction algorithm, the reader is urged to consult~\cite{groot,duff2002olc}.

\section{Experimental Evaluation}
\label{sec:experiments}
We perform three sets of experiments. The first set investigates the performance of \HIPER{} with various choices of the parameter $\delta$, including the optimal choice suggested by Theorem~\ref{theorem:hiper}. The second set compares \HIPER{} with the {\em myopic} and {\em optimistic} approximations. In the final set of experiments, we compare the {\em optimistic} with the {\em finite lookahead}  approximation. In all cases, we collected results from $10^4$ runs, with $100$ nodes in each simulation, and we plot a moving average of the {\em expected loss} as various network parameters change.
\begin{figure*}[ht]
  \centering
  \subfigure[Horizon]{
    \psfrag{regret}[B][B][1][0]{$\Ex L$}
    \psfrag{d1}[B][B][1][0]{$\delta_1$}
    \psfrag{d2}[B][B][1][0]{$\delta_2$}
    \psfrag{d3}[B][B][1][0]{$\delta_3$}
    \psfrag{dr}[B][B][1][0]{$\delta^*$}
    \psfrag{horizon}[B][B][1][0]{$\Ho$}
    \includegraphics[width=0.45\textwidth]{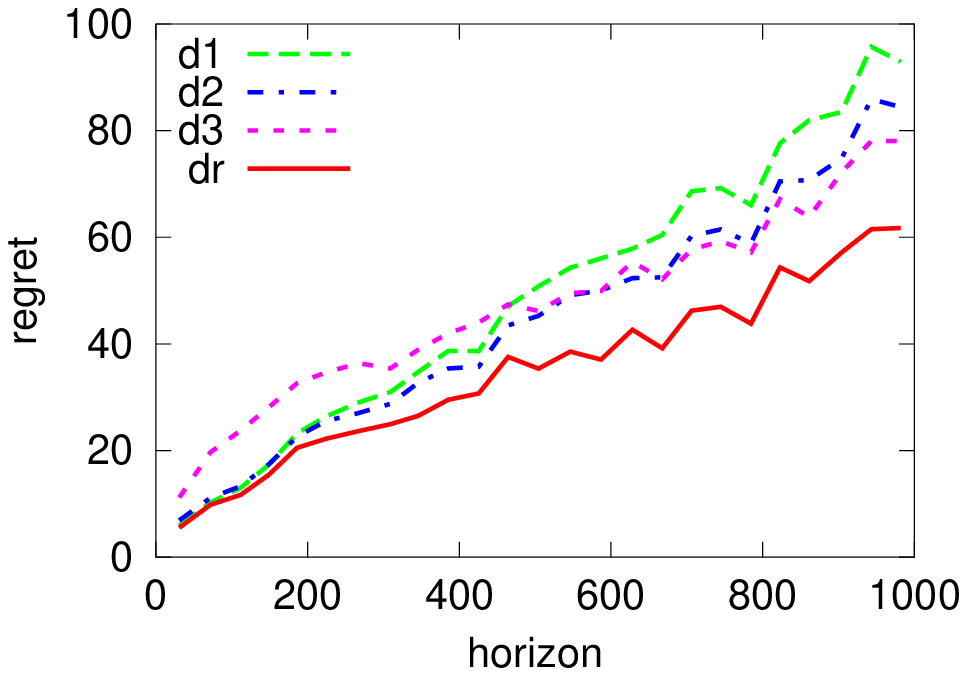}
  }
  \subfigure[Gap]{
    \psfrag{regret}[B][B][1][0]{$\Ex L$}
    \psfrag{d1}[B][B][1][0]{$\delta_1$}
    \psfrag{d2}[B][B][1][0]{$\delta_2$}
    \psfrag{d3}[B][B][1][0]{$\delta_3$}
    \psfrag{dr}[B][B][1][0]{$\delta^*$}
    \psfrag{delta}[B][B][1][0]{$\Delta$}
    \includegraphics[width=0.45\textwidth]{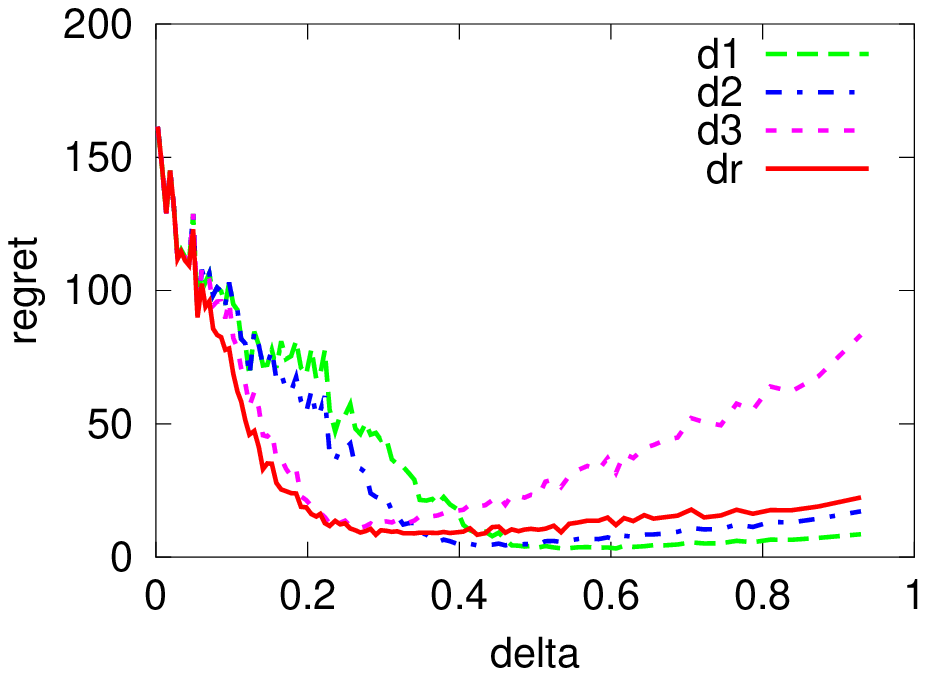}
  }
  \\
  \subfigure[Malicious]{
    \psfrag{regret}[B][B][1][0]{$\Ex L$}
    \psfrag{d1}[B][B][1][0]{$\delta_1$}
    \psfrag{d2}[B][B][1][0]{$\delta_2$}
    \psfrag{d3}[B][B][1][0]{$\delta_3$}
    \psfrag{dr}[B][B][1][0]{$\delta^*$}
    \psfrag{malicious}[B][B][1][0]{proportion of malicious nodes}
    \includegraphics[width=0.45\textwidth]{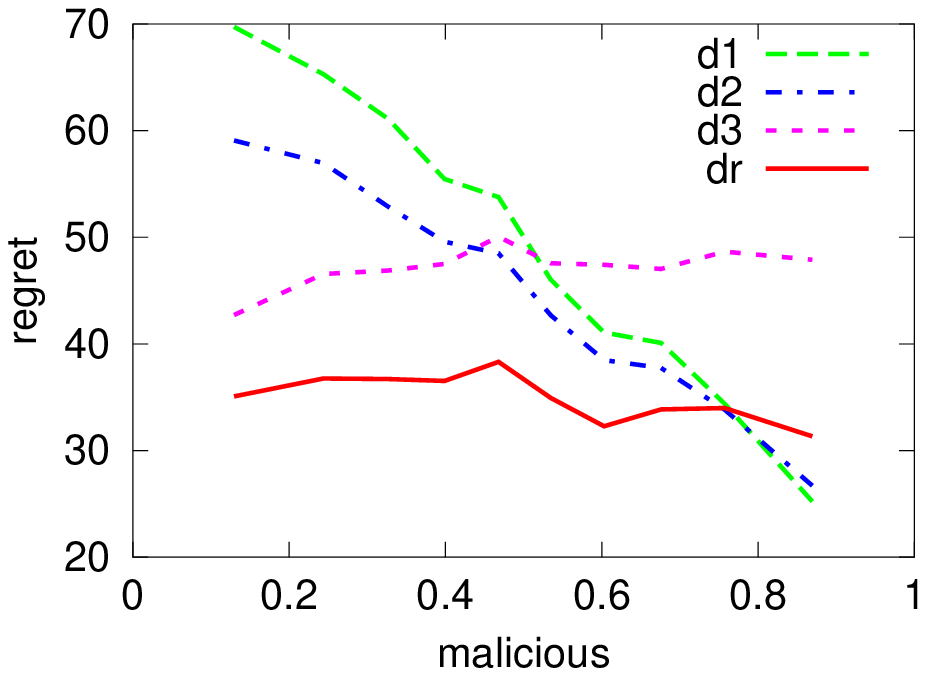}
  }
  \subfigure[Gain]{
    \psfrag{regret}[B][B][1][0]{$\Ex L$}
    \psfrag{d1}[B][B][1][0]{$\delta_1$}
    \psfrag{d2}[B][B][1][0]{$\delta_2$}
    \psfrag{d3}[B][B][1][0]{$\delta_3$}
    \psfrag{dr}[B][B][1][0]{$\delta^*$}
    \psfrag{payment}[B][B][1][0]{$\GU$}
    \includegraphics[width=0.45\textwidth]{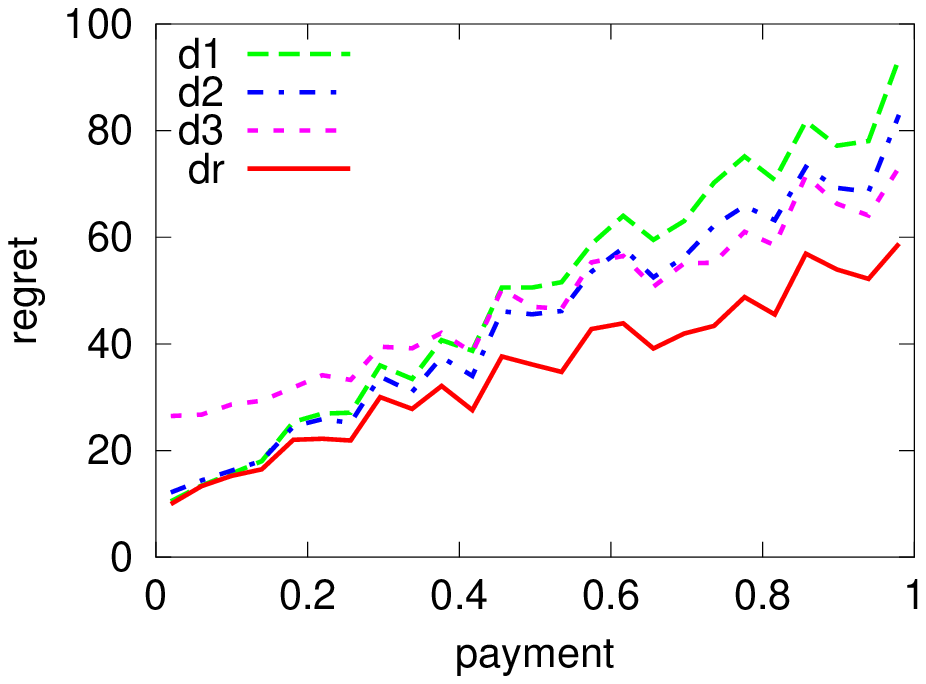}
  }
  \caption{Simulations with Alg.~1, for four different choices of $\delta$. In particular $\delta_1 = 0.9$, $\delta_2 = 0.95$, $\delta_3 = 0.99$ and $\delta^*$ is chosen according to Theorem~1. It can be seen that, while the algorithm is not extremely sensitive to the exact choice of $\delta$, the optimal value is generally more robust.}
  \label{fig:simulations}
\end{figure*}
Specifically, the first results we report (i.e. Fig. \ref{fig:simulations})  are made through $10^4$ experiments. For each experiment, we selected a horizon $\Ho \sim Uniform([10,1000])$, user and adversary parameters $u, q \sim Uniform([0,1])$, and user gain $\GU \sim Uniform([0,1])$ and we set $\LQ = 1$. Each experiment measured the loss for a network containing 100 nodes, each of which had a probability $p$ of being malicious, with $p \sim Beta(2,2)$ for each experiment. During each run, the $i$-th node generates a sequence of observations $x_{i,t}$ drawn from a Bernoulli distribution with parameter $u$ if the node is honest and $q$ if the node is malicious. Figure~\ref{fig:simulations} shows a summary of the results, averaged over these trials. It can be seen that, while \HIPER{}'s performance is relatively robust to the choice of $\delta$, nevertheless the optimal choice suggested by Theorem~\ref{theorem:hiper} generally leads to small losses.
\begin{figure*}[ht]
  \centering
  \subfigure[Horizon]{
    \psfrag{dr}[r][r][0.8][0]{\HIPER{}}
    \psfrag{myopic}[r][r][0.8][0]{myopic}
    \psfrag{optimistic}[r][r][0.8][0]{optimistic}
    \psfrag{regret}[B][B][1][0]{$\Ex L$}
    \psfrag{horizon}[B][B][1][0]{$\Ho$}
    \includegraphics[width=0.45\textwidth]{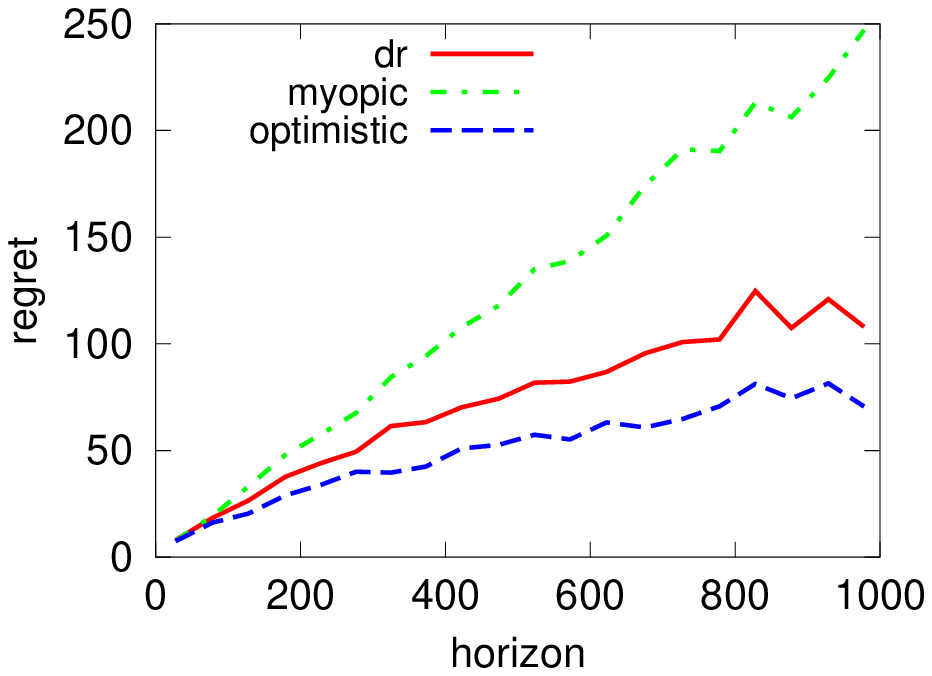}
  }
  \subfigure[Gap]{
    \psfrag{dr}[r][r][0.8][0]{\HIPER{}}
    \psfrag{myopic}[r][r][0.8][0]{myopic}
    \psfrag{optimistic}[r][r][0.8][0]{optimistic}
    \psfrag{regret}[B][B][1][0]{$\Ex L$}
    \psfrag{delta}[B][B][1][0]{$\Delta$}
    \includegraphics[width=0.45\textwidth]{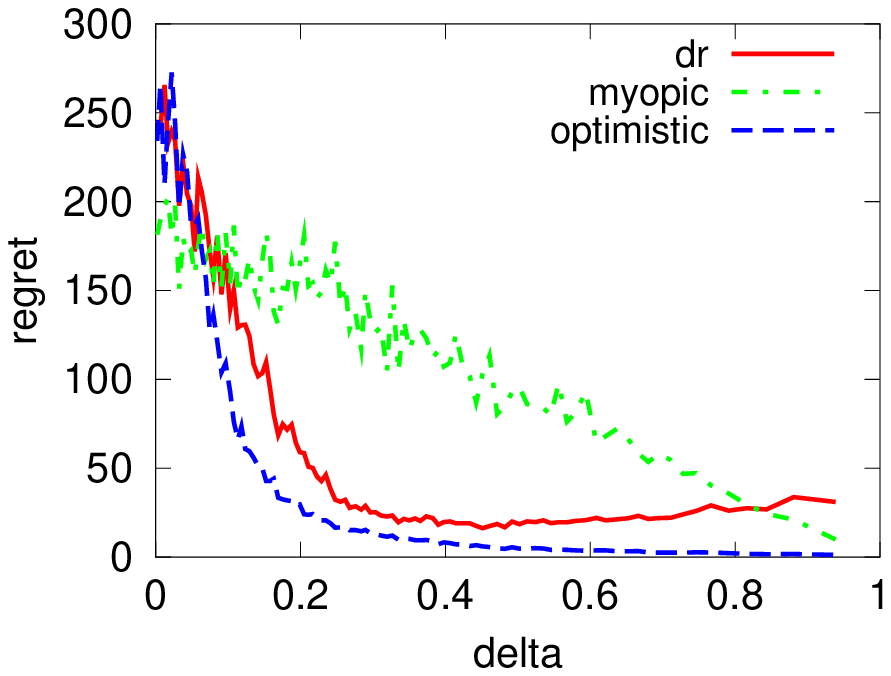}
  }
  \\
  \subfigure[Malicious]{
    \psfrag{dr}[r][r][0.8][0]{\HIPER{}}
    \psfrag{myopic}[r][r][0.8][0]{myopic}
    \psfrag{optimistic}[r][r][0.8][0]{optimistic}
    \psfrag{regret}[B][B][1][0]{$\Ex L$}
    \psfrag{malicious}[B][B][1][0]{proportion of malicious nodes}
    \includegraphics[width=0.45\textwidth]{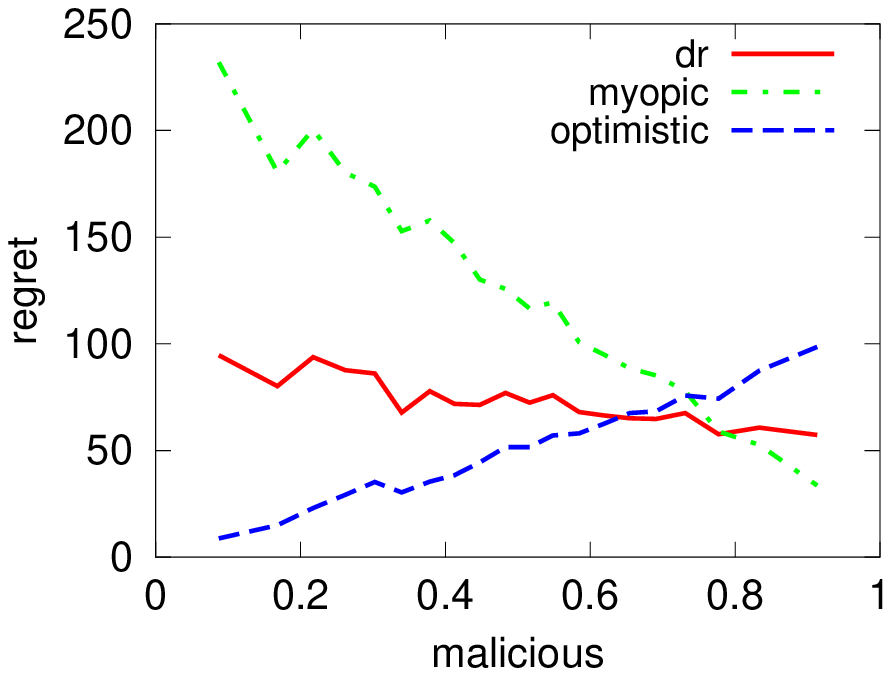}
  }
  \subfigure[Gain]{
    \psfrag{dr}[r][r][0.8][0]{\HIPER{}}
    \psfrag{myopic}[r][r][0.8][0]{myopic}
    \psfrag{optimistic}[r][r][0.8][0]{optimistic}
    \psfrag{regret}[B][B][1][0]{$\Ex L$}
    \psfrag{payment}[B][B][1][0]{$\GU$}
    \includegraphics[width=0.45\textwidth]{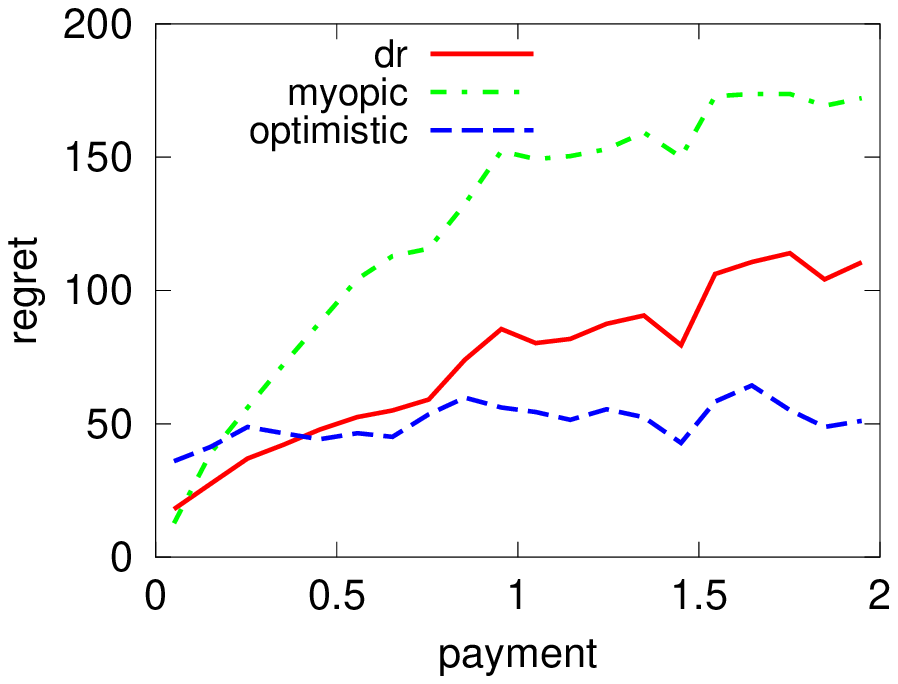}
  }
  \caption{Comparison of \HIPER{} with the {\em myopic} solver and the {\em optimistic} approximation for various network conditions. It can be clearly seen that the {\em myopic} approximation is significantly worse than both approaches. However, the {\em optimistic} approach outperforms the worst-case \HIPER{} algorithm when the proportion of malicious nodes is low. The {\em optimistic} approach is also better when the payment for honest nodes is high.}
  \label{fig:pomdp_simulations}
\end{figure*}

For our second set of experiments, shown in Figure~\ref{fig:pomdp_simulations}, we compare \HIPER{} with the {\em optimistic} and {\em myopic}  algorithms. We increased the range of user gains to $\GU \sim Uniform([0,2])$ compared to the previous setup, but the other experimental parameters remain the same. It is clear that the {\em myopic} approximation has almost always a higher loss compared to both \HIPER{} and the {\em optimistic} algorithm. The latter, while performing at a similar level to \HIPER{}, has an advantage when either the proportion of malicious is small or when $\GU$ is large. This makes sense intuitively, since in those cases the optimism is justified. In the converse case, however, the {\em optimistic} approach performs worse than \HIPER{}, which is less sensitive to the proportion of malicious nodes, since it is a worst-case approach.
\begin{figure*}[ht]
  \centering
  \subfigure[Horizon]{
    \psfrag{optimistic}[r][r][0.8][0]{optimistic}
    \psfrag{finite}[r][r][0.8][0]{$T=4$}
    \psfrag{finite2}[r][r][0.8][0]{$T=8$}
    \psfrag{regret}[B][B][1][0]{$\Ex L$}
    \psfrag{horizon}[B][B][1][0]{$\Ho$}
    \includegraphics[width=0.45\textwidth]{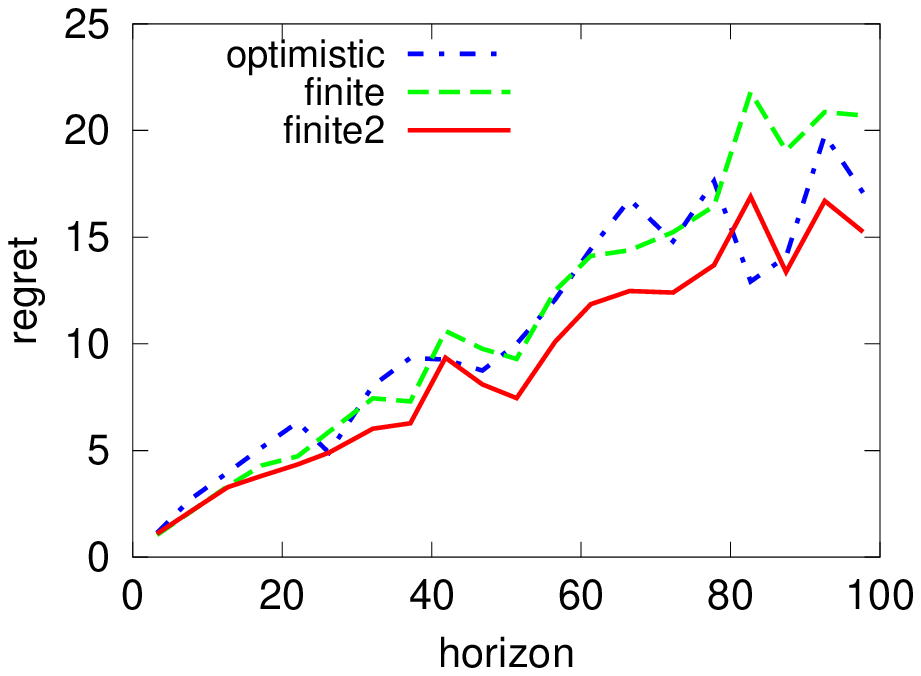}
  }
  \subfigure[Gap]{
    \psfrag{optimistic}[r][r][0.8][0]{optimistic}
    \psfrag{finite}[r][r][0.8][0]{$T=4$}
    \psfrag{finite2}[r][r][0.8][0]{$T=8$}
    \psfrag{regret}[B][B][1][0]{$\Ex L$}
    \psfrag{delta}[B][B][1][0]{$\Delta$}
    \includegraphics[width=0.45\textwidth]{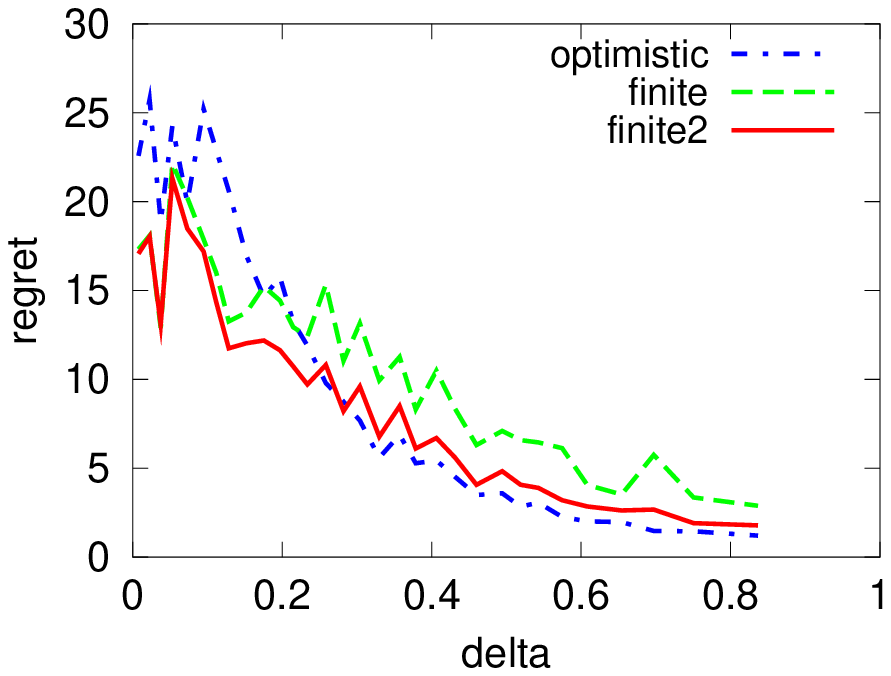}
  }
  \\
  \subfigure[Malicious]{
    \psfrag{optimistic}[r][r][0.8][0]{optimistic}
    \psfrag{finite}[r][r][0.8][0]{$T=4$}
    \psfrag{finite2}[r][r][0.8][0]{$T=8$}
    \psfrag{regret}[B][B][1][0]{$\Ex L$}
    \psfrag{malicious}[B][B][1][0]{proportion of malicious nodes}
    \includegraphics[width=0.45\textwidth]{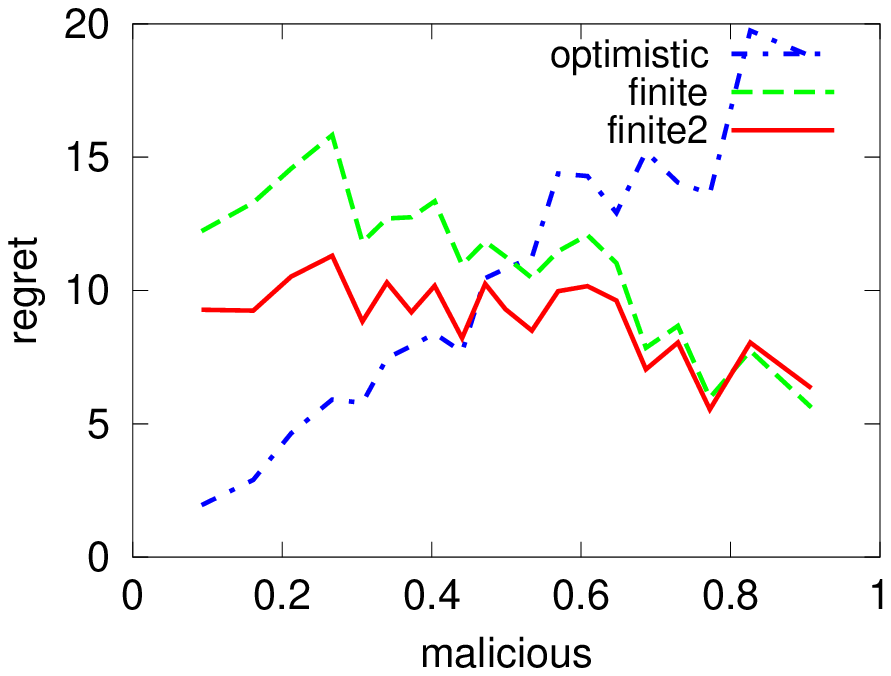}
  }
  \subfigure[Gain]{
    \psfrag{optimistic}[r][r][0.8][0]{optimistic}
    \psfrag{finite}[r][r][0.8][0]{$T=4$}
    \psfrag{finite2}[r][r][0.8][0]{$T=8$}
    \psfrag{regret}[B][B][1][0]{$\Ex L$}
    \psfrag{payment}[B][B][1][0]{$\GU$}
    \includegraphics[width=0.45\textwidth]{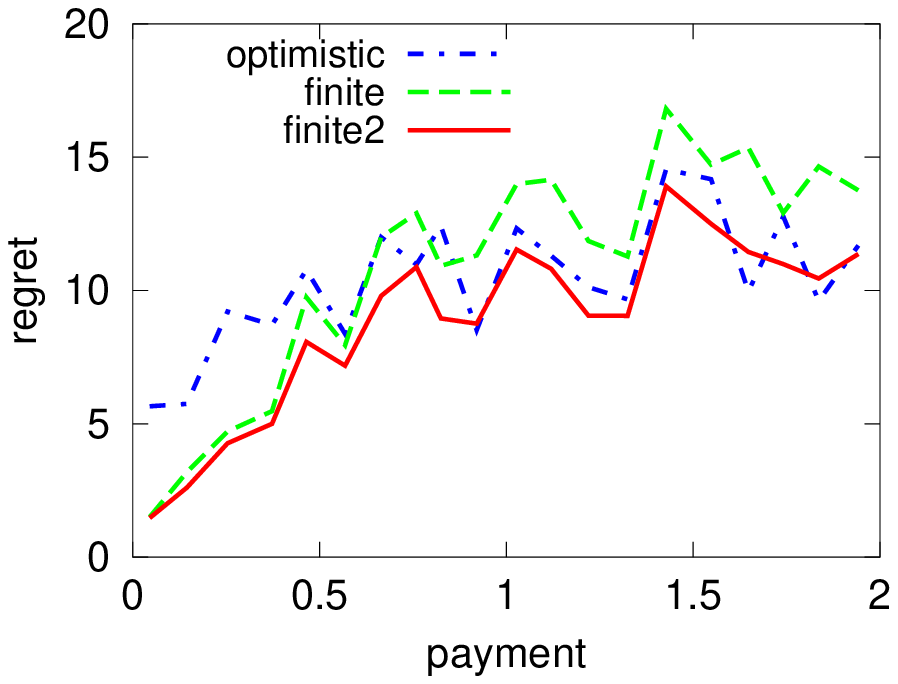}
  }
  \caption{Comparison of the optimistic approximation with approximate {\em non-myopic} POMDP solvers for planning lookahead of $T$ time-steps where $T \in \{4, 8\}$. It can be seen that, for short horizons, these perform just as well  and that they are more robust to the proportion of malicious nodes in the network. However, these methods are computationally more intensive, with complexity $O(e^T)$.}
  \label{fig:fh_simulations}
\end{figure*}

Finally, we performed some experiments comparing the {\em optimistic} approximation with the {\em finite-lookahead} POMDP solvers for lookahead for $T$ time-steps where $T \in \{4, 8\}$. While these do not solve the problem to the end of the horizon $H$, they plan ahead for $T$ steps at every time-step of the simulation. Unfortunately, the complexity of these solvers is exponential in $T$, which limited the amount of simulations we could perform to $10^3$ and we only considered horizons $\Ho \sim Uniform([1,100])$. These experiments are shown in Fig.~\ref{fig:fh_simulations}. In comparison with Fig.~\ref{fig:pomdp_simulations}, the {\em finite lookahead} algorithms performs much better than the {\em myopic} approximation and indeed the $8$-step lookahead manages to slightly outperform the {\em optimistic} approximation. In addition, it is much more robust to the proportion of malicious nodes in the network. However, the relative advantage of the 8-step to the 4-step lookahead is relatively small for the amount of extra computation required.\footnote{The computational effort is exponential in $T$.}

\section{Conclusion}
\label{sec:conclusion}

This paper defined a resource management problem that arises frequently in communication networks. Namely, whether to remove a suspicious node from the system, with the amount of available evidence, or to collect some further data before taking the final decision. This is in fact a type of stopping problem, which we believe is of relevance to many applications where blacklisting may be performed. This includes applications such as automated intrusion response, as well as ensuring fairness in peer-to-peer networks, such as~\cite{vieira2009fighting}. To this end, we proposed and analysed, both theoretically and experimentally, an efficient algorithm, \HIPER{}, that achieves low {\em worst-case expected loss} relative to an oracle that knows {\em a priori} the type (honest or malicious) of every node in the system. In addition, we derived and compared a number of algorithms by modelling the problem as a POMDP: a {\em myopic} and an {\em optimistic} approximation, as well as a {\em finite lookahead} solver. Of those, the {\em optimistic} approximation and the partial {\em finite lookahead} solvers perform the best, with the {\em finite lookahead} methods being the most robust, while simultaneously being computationally demanding. 

The main advantage of \HIPER{} are its simplicity and lack of stringent assumptions on the distribution. This makes it suitable for deployment in most situations. However, whenever a full probabilistic model and computational resources are available, one of the approximate solvers would be useful.
The overall best performance is offered by the {\em finite lookahead}, closely followed by the {\em optimistic} approximation. The {\em myopic} approximation, which is equivalent to the widely-used ``most likely state'' (MLS) approximation, is the worst. To our knowledge, neither the {\em optimistic} approximation, nor the {\em finite lookahead} methods have been applied before to this problem or more generally to intrusion response problems. They should be more generally applicable for other types of intrusion response and resource management problems. It is our view that they are inherently more suitable than other approximations such as the commonly used (MLS) approximation (or equivalently, a sequential probability ratio test) which in our setting produces an essentially random policy.

For future work, we would like to extend our theoretical analysis to the performance of the  {\em optimistic} and the {\em finite lookahead} algorithms. In addition, it would be interesting to examine a more general game-theoretic scenarios, including strategic attackers~\cite{bao-sec-class-game,dritsoula-game-class-game}. Finally, we would like to generalise our setting so that observations must be {\em explicitly} gathered from each node, where it is not possible to continuously sample all nodes due to budget constraints. In fact, the sampling problem in the context of intrusion detection, has been recently studied by~\cite{Liu:DID,bu2011structural}. A natural extension of our work would consequently be to optimally combine sampling and response policies.

\
\appendix
\section{Proofs}
\label{app:proofs}
This section collects the missing proofs from the main text.
\begin{proof}[(Lemma~\ref{lem:lossattacker})]
Since the node $i$ under consideration is malicious, i.e. $i\in \Q$, it holds that: $\Ex[x_{i,t}\mid \Q]=q$.
Then, we have:
\begin{equation*}
\Ex[\theta_t\mid \Q]=\Ex\bigg[\frac{1}{t}\cdot \sum_{k=1}^{t} x_{i,k} \, \bigg | \,  \Q \bigg]= \frac{1}{t} \sum_{k=1}^{t} \Ex[x_{i,k}\mid \Q]= \frac{1}{t}\cdot t \cdot q=q.
\end{equation*}
%
From Hoeffding's inequality ({\em Lemma} \ref{lem:hoeffding}, in the Appendix), we have:
%
\begin{equation}
\Pr\left(|\theta_t - q| > \epsilon_t\mid \Q \right)\leq 2 \exp(-2t\epsilon_t^2),
\label{eq:hoefdingStop}
\end{equation}
where $\epsilon_t>0$ and
$\Pr\left(|\theta_t - q| > \epsilon_t \mid \Q \right)$ denotes the probability that $\theta_t$ (which is random) is very far away from $q$ (which is fixed).
Now set:
$\epsilon_t=\sqrt{\frac{\ln(2/\delta)}{2t}}$
as in Algorithm 1. Then, since equation \ref{eq:hoefdingStop} holds for any $\epsilon_t>0$, we get that the probability of keeping a malicious node $i\in \Q$ in the network is at most $\delta$:
$\Pr \bigg( | \theta_t - q  | > \sqrt{\frac{\ln(2/\delta)} {2t}} \mid \Q \bigg) < \delta$.
Thus, we have:
$\Ex[L\mid \Q]=\Ex[N \mid \Q ] \cdot \ell_\Q=\sum_{t=0}^{\infty}\Pr(N=t \mid \Q) \cdot t \cdot \ell_\Q
\leq \ell_\Q \sum_{t=0}^{\infty} \delta^{t-1}\cdot t = \frac{\ell_\Q}{(1-\delta)^2} $
%
%
%
\end{proof}
\begin{proof}[(Proof of Lemma~\ref{lem:lossuser})]
  We denote by $N$ the time-step at which \E\ removes node $i$ from the network.
  Then, the function $g: \mathbb{N}^2\rightarrow \mathbb{R}$ that gives us the gain for each node $i$ is defined as:
    $g(n,h)\defn \min \{n,h\}\cdot \GU$
  where $h\in\Ho$ and $n\in N$.
  Since the node $i$ under consideration is honest, i.e. $i\in \U$, we have $\Ex[x_{i,t}\mid \U]=u$.
  Without loss of generality we assume that: $u=q+\Delta$, where $\Delta>0$. So we only need  $\Pr(\theta_t-q<\epsilon_t\mid \U)$.
  Since $q=u-\Delta$ from the Hoeffding inequality ({\em Lemma} \ref{lem:hoeffding}, in the Appendix), we have:
  %
  \begin{align*}
    \Pr(N=t\mid \U) \leq \Pr(N\leq t)&\leq  \Pr(\theta_t-u<\epsilon_t-\Delta\mid \U)\\&\leq \exp(-2\cdot t (\epsilon_t-\Delta)^2)
    \label{eq:UserN=T3}
  \end{align*}
  where $\Delta - \epsilon_t>0$.
  It holds that:
  \begin{align}
    &\Ex[G\mid \U, N=n]=\nonumber
    \\&\sum_{n=0}^{\infty}\Pr(\Ho=h\mid \U,N=n)\Ex[ G\mid \U,N=n,\Ho=h]
  \end{align}
  But it holds that:
  $\Ex[G\mid \U, N=n, \Ho]=g(n,h)$
  and since $h\in\Ho$ and $n\in N$ are independent we have:
 $ \Pr(\Ho=h\mid \U,N=n)=\Pr(\Ho=h\mid \U).$
  Thus,
  \begin{align}
    &\Ex[G\mid \U,N=n]=\sum_{h=0}^\infty \Pr(\Ho=h\mid \U)\cdot g(n,h)=\nonumber\\&\sum_{h=0}\Pr(\Ho=h\mid \U)\min\{n,h\} \cdot \GU=\nonumber\\
    &\GU \cdot \Big\{ \sum_{h=0}^{n-1} \Pr(\Ho=h\mid \U)\cdot h +  \sum_{h=n}^{\infty} \Pr(\Ho=h\mid \U)\cdot n\Big\}
  \end{align}
  The expected loss is given by subtracting from the expected gain of the oracle policy, when \E\ never removes the node from the network (i.e. $N=\infty$), the expected gain when \E\ removes the node at the time-step $N=n$. Thus, it holds:
  \begin{align}
    &\Ex[L\mid \U, N=n]=
    \Ex[G\mid \U,N=\infty] - \Ex[G \mid \U, N=n]=\nonumber\\
    &\lim_{n\rightarrow\infty}(\Ex[G\mid\U, N=n])-\Ex[G\mid \U, N=n]=\nonumber\\
    &\GU \sum_{h=0}^{\infty} \Pr(\Ho=h) h-\GU \Big\{  \sum_{h=0}^{n-1} \Pr(\Ho=h) h + \sum_{h=n}^{\infty} \Pr(\Ho=h) n  \Big\}   \nonumber\\
    &=\GU \Big\{\sum_{h=n}^{\infty}\Pr(\Ho=h) h -\sum_{h=n}^{\infty} \Pr(\Ho=h) n \Big\}
  \end{align}
  Since, by definition $\Pr(\Ho=h+1\mid \Ho>h )=\lambda$, we have
  $\Pr(\Ho=h)= (1-\lambda)^{h-1} \lambda$. Consequently,
  \begin{align}
    \Ex[L\mid \U, N=n]=& \GU \lambda \bigg( \sum_{h=n}^{\infty}(1-\lambda)^{h-1}  h - \sum_{h=n}^\infty (1-\lambda)^{h-1} n\bigg )\nonumber\\
    =& \GU    \frac {(1-\lambda)^n} {\lambda}
  \end{align}
  Thus, we have:
  $  \Ex[L\mid \U]= \sum_{t=0}^{\infty}\Pr(N=t\mid \U) \Ex[L|N=t]\nonumber\\
    \leq \sum_{t=0}^\infty \exp(-2\cdot t\cdot (\epsilon_t -\Delta)^2)\cdot \GU \frac {(1-\lambda)^t} {\lambda}$
  Since the algorithm uses $\epsilon_t=\frac{\Delta} {\sqrt{t}}$, we have:    
  \begin{align}
    E[L\mid \U] \leq& \frac {\GU} {\lambda}\sum_{t=0}^\infty \exp\left(-2\cdot t \left[\frac {\Delta} {\sqrt{t}} -\Delta\right]^2  \right)(1-\lambda)^t\nonumber \\
    &= \frac {\GU} {\lambda}\sum_{t=0}^\infty \exp \left(-2 \Delta^2 (\sqrt{t}-1)^2\right)(1-\lambda)^t\nonumber \\
    &\leq \frac {\GU} {\lambda}\sum_{t=0}^\infty \exp \left(-2 \Delta^2 \left[\sqrt{t}-\sqrt{\frac {t} {2}}\right]^2\right)(1-\lambda)^t \nonumber\\
    &= \frac  {\GU} {\lambda}\sum_{t=0}^\infty \left [\exp \left( -\frac {\Delta^2}{2}\right)(1-\lambda)\right]^t\nonumber \\
    &=   \frac {\GU} {\left[1- \exp \left( -\frac {\Delta^2}{2}\right)(1-\lambda)\right]\lambda}\nonumber \leq \frac {\GU(\Delta^2+2)} {\lambda (\Delta^2+2\lambda)} 
  \end{align}         
  where $t\geq2$.
\end{proof}

\section{Additional results}
\label{app:additional}

\begin{definition}[Bernoulli distribution]
  If $X_1, \ldots, X_n$ are independent Bernoulli random variables
  with $X_k \in \set{0, 1}$ and $\Pr(X_k = 1) = \mu$ for all $k$, then
  \begin{equation}
    \label{eq:binomial}
    \Pr\left(\sum_{k=1}^n X_k \geq u\right)=
    \sum_{k=0}^{u} \binom{n}{k} \mu^k (1 - \mu)^{n-k}.
  \end{equation}
\end{definition}
\label{sec:Hoeffding}
\begin{lemma}[Hoeffding]
  For independent random variables $X_1, \ldots, X_n$ such that
  $X_i \in [a_i, b_i]$, with  $\mu_i \defn  \Ex X_i$ and $t>0$:
  \begin{align*}
    \Pr\left(
      \sum_{i=1}^n X_i \geq \sum_{i=1}^n \mu_i + n t
    \right) 
    &\leq
    \exp\left(
      -\frac{2n^2t^2}{\sum_{i=1}^n (b_i - a_i)^2}
    \right).
  \end{align*}
  \label{lem:hoeffding}
  The same in equality holds for $\sum_{i=1}^n X_i \leq \sum_{i=1}^n \mu_i - n t$.
\end{lemma}

\bibliographystyle{abbrv}
\bibliography{response}


\end{document}